\documentclass[sigconf]{acmart}

\AtBeginDocument{%
  \providecommand\BibTeX{{%
    \normalfont B\kern-0.5em{\scshape i\kern-0.25em b}\kern-0.8em\TeX}}}

\copyrightyear{2021}
\acmYear{2021}
\setcopyright{rightsretained}
\acmConference[KDD '21]{Proceedings of the 27th ACM SIGKDD Conference on Knowledge Discovery and Data Mining}{August 14--18, 2021}{Virtual Event, Singapore}
\acmBooktitle{Proceedings of the 27th ACM SIGKDD Conference on Knowledge Discovery and Data Mining (KDD '21), August 14--18, 2021, Virtual Event, Singapore}\acmDOI{10.1145/3447548.3467402}
\acmISBN{978-1-4503-8332-5/21/08}
\usepackage[utf8]{inputenc} 
\usepackage[T1]{fontenc}    
\usepackage{hyperref}       
\usepackage{url}            
\usepackage{booktabs}       
\usepackage{amsfonts}       
\usepackage{nicefrac}       
\usepackage{microtype}      
\usepackage{amsmath,amsthm,amsfonts,amscd} 
\usepackage{mathtools}
\usepackage{enumerate}
\usepackage{fancyhdr}
\usepackage{mathrsfs}
\usepackage{listings}
\usepackage{subcaption}
\usepackage[english]{babel}
\usepackage{apxproof}
\usepackage{thmtools}
\usepackage{thm-restate}
\usepackage{enumitem}
\usepackage{float}
\usepackage{bm}
\usepackage{wrapfig}
\usepackage{algorithm}

\usepackage{algpseudocode}

\theoremstyle{definition}
\newtheorem{definition}{Definition}[section]

\newcommand{\R}{\mathbb{R}}
\newcommand{\mat}[1]{\bm{#1}}

\newcommand{\W}{\mathcal{W}}

\definecolor{myGreen}{rgb}{0,0.5,0}

\DeclareMathOperator{\argmin}{argmin}

\begin{document}

\fancyhead{}

\title{Fine-Grained System Identification of Nonlinear Neural Circuits}

\author{Dawna Bagherian}
\affiliation{%
  \institution{California Institute of Technology}
  \city{Pasadena, CA}
  \country{USA}}
\email{dawna@caltech.edu}

\author{James Gornet}
\affiliation{%
  \institution{California Institute of Technology}
  \city{Pasadena, CA}
  \country{USA}}
\email{jgornet@caltech.edu}

\author{Jeremy Bernstein}
\affiliation{%
  \institution{California Institute of Technology}
  \city{Pasadena, CA}
  \country{USA}}
\email{bernstein@caltech.edu}

\author{Yu-Li Ni}
\affiliation{%
  \institution{California Institute of Technology}
  \city{Pasadena, CA}
  \country{USA}}
\email{ynni@caltech.edu}

\author{Yisong Yue}
\affiliation{%
  \institution{California Institute of Technology}
  \city{Pasadena, CA}
  \country{USA}}
\email{yyue@caltech.edu}

\author{Markus Meister}
\affiliation{%
  \institution{California Institute of Technology}
  \city{Pasadena, CA}
  \country{USA}}
\email{meister@caltech.edu}

\renewcommand{\shortauthors}{Bagherian, et al.}


\begin{abstract}
We study the problem of sparse nonlinear model recovery of high dimensional compositional functions.  Our study is motivated by emerging opportunities in neuroscience to recover fine-grained models of biological neural circuits using collected measurement data.  Guided by available domain knowledge in neuroscience, we explore conditions under which one can recover the underlying  biological circuit that generated the training data. 
Our results suggest insights of both theoretical and practical interests.
Most notably, we find that a sign constraint on the weights is a necessary condition for system recovery, which we establish both theoretically with an identifiability guarantee and empirically on simulated biological circuits.  We conclude with a case study on retinal ganglion cell circuits using data collected from mouse retina, showcasing the practical potential of this approach.  
\end{abstract}

\ccsdesc[300]{Networks~Network structure}
\ccsdesc[300]{Applied computing~Systems biology}
\ccsdesc[500]{Applied computing~Biological networks}
\ccsdesc[300]{Computing methodologies~Modeling methodologies}
\ccsdesc[300]{Computing methodologies~Neural networks}

\keywords{nonlinear system identification, neural networks, neuroscience}


\maketitle

\section{Introduction}

Sparse system estimation or identification is a core problem in science and engineering \cite{hastie2015statistical,ljung1999system,mccarter2014sparse,zheng2018dags}, where the goal is to recover from training data a sparse set of non-zero parameters in a high-dimensional function class.  Examples include estimating dynamical systems \cite{ljung1999system,fry2011sign}, graphical models or inverse covariance matrices \cite{mccarter2014sparse,zheng2018dags}, and linear effects models \cite{hastie2015statistical}.  Thus far, the bulk of prior work has focused on  linear models, which have limited applicability when the underlying system is inherently nonlinear.



In many domains there is a growing need for reliably estimating a sparse nonlinear system from data. We are particularly motivated by applications to circuit neuroscience, where the goal is to estimate the fine-grained structure and connection weights of a neural circuit.  A typical data collection approach is to apply stimuli to the sensory neurons in the input layer of the circuit, and measure the responses of the output neurons \cite{meister1994mea}. 
%
Recent advancements in biotechnology have yielded increasingly powerful methods for collecting neuronal data, including genetic tools for labeling neurons \cite{jo2018amgclabel}, dense electrode arrays for recording action potentials from living cells \cite{jun2017neuropixel}, genetically encoded fluorescent activity reporters \cite{chen2013gcamp6}, and complex surgical procedures to provide access to specific brain regions \cite{feinberg2014sc}.  

Motivated by the neuroscience modeling application, we aim to estimate a sparse nonlinear model from a high-dimensional compositional function class, i.e., a deep neural network.  The key question we ask is: \textit{under what conditions can one recover the parameters of a sparse neural network?} To keep our inquiry practically grounded, we restrict to exploring conditions that abstractly capture conventional neuroscience domain knowledge. 

In this paper, we study conditions for successful sparse parameter recovery of a nonlinear composition function with skip connections.  Notably, we show that having each layer of weights be sign constrained is a necessary condition (in addition to other more standard necessary conditions).  We establish this result both theoretically with an identifiability guarantee for a specific class of neural networks, and empirically on a broader class of simulated biological circuits.   Our experiments show a striking dichotomy: $\ell_1$-regularized regression fails to recover the correct parameters without sign constraints, and succeeds with them.  Our results also provide guidance as to the dependence of system identification on the training dataset, the structure of the circuit under study, and the quantity of domain foreknowledge provided to the neural network.

 
 We conclude with a case study of the mouse retina. We measure responses to visual input from four types of retinal ganglion cells in the live mouse retina using a multi-electrode array \cite{meister1994mea}. We confirm the current scientific understanding of the structure of one of these four and suggest an additional structural component. We also provide a hypothesis as to the structure of the other three types, which can later be confirmed with biological experiments. These results showcase the practical potential of this approach.
\section{Related work}

\textbf{System identification and structure recovery.} Broadly speaking, system identification is the use of statistical methods to build models of systems from measured data \cite{ljung1999system}.  System identification is a key tool in modeling dynamical systems, which includes early work on neural system identification for control systems \cite{kuschewski1993application}.  A fundamental issue that arises is \textit{identifiability} \cite{bellman1970structural}---that is, when can one uniquely recover the true system. This affects the reliability of the results for downstream scientific analysis. 

Identifiability has been studied theoretically for nonlinear neural networks \cite{sussman,fefferman,albertini,helmut,rolnick,schroderneurips}, ``linear networks'' in the context of matrix factorization \cite{donoho_when_2004,laurberg_theorems_2008,nnmf_revisited}, and linear autoregressive models \cite{fry2011sign}. In Section \ref{sec:theory-result}, we introduce and discuss some of the theoretical concepts that have a bearing on identifiability.  As stated before, our main goal is to provide a thorough theory-to-practice investigation grounded in real neuroscience modeling challenges.


A related concept is support recovery, where the main goal is to discover which parameters of a model are non-zero \cite{tibshirani1996regression,hastie2015statistical,shen2017iteration, somani2018support,li2015sparsistency}. Support recovery can be thought of as a subgoal of full system identification.
It is commonly studied in sparse linear systems that have few non-zero parameters.  While biological neural networks are also sparse \cite{gollisch2010, mizrahi2013}, they are nonlinear multi-layer models for which theoretical results in sparse linear support recovery do not directly apply.  
Nonetheless, we show that, under suitable conditions, one can employ $\ell_1$-regularized regression (that is commonly used for estimating sparse linear models \cite{tibshirani1996regression}) to reliably estimate sparse neural networks with limited training data; an interesting future direction would be to establish \textit{sparsistency} guarantees \cite{li2015sparsistency}.

Another related concept is structure discovery in (causal) graphical models \cite{zheng2018dags,mccarter2014sparse,sohn2012joint,saeed2020causal,tillman2014learning}.  A typical setting is to recover the structure of a directed acyclic network (i.e., which edges are non-zero) that forms the causal or generative model of the data.  This setting is very similar to ours with a few differences.  First, the goal of structure discovery in graphical models is to recover the direction of the edges in addition to the weights, whereas in our setting all the edge directions are known a priori.  Second, the training data for structure discovery is typically fully observed in terms of measuring every node in the network, whereas for our setting we only observe the inputs and final outputs of the network (and not the measurements of nodes in the hidden layers).  Like in sparse system identification, most prior work in structure discovery of graphical models is restricted to the linear setting.

\textbf{Interplay between neuroscience and deep learning.} The interplay between neuroscience and deep learning has a rich history that ebbs and flows between one field driving progress in the other.  For instance, convolutional neural networks (CNNs) are originally inspired by biology \cite{fukushima1980neocognitron}, and were designed to have computational characteristics of the vertebrate visual system. More recently, researchers have endeavored to understand biological, and specifically neuronal, systems using artificial neural networks \cite{mcintosh2016retina, tanaka2019retina, klindt2017neural,seeliger2019end,mcintosh2015deep,abbasi2018deeptune,yamins2014performance, schroderneurips}. Such work has been largely restricted to coarse-grained analyses that characterize computation of regions of the brain, rather than that of individual neurons interacting within a circuit, or to the revelation of microcircuit motifs within trained CNNs that mirror those found in biology. 

In contrast, we are interested in fine-grained circuit modeling where the neurons in the learned neural network have a one-to-one correspondence with individual biological units. Preliminary efforts in this direction have been made for a smaller model of one layer of retinal synapses \cite{schroderneurips}, but to the best of our knowledge, this work covers only parameter estimation, not structure estimation, and the general problem has not been studied theoretically or empirically via systematic simulations, nor has it been tested on real data from deeper circuits.

\section{Problem Statement}
\label{sec:problem}

We formulate the problem of system identification in nonlinear feed-forward networks as follows. Consider a function $f(\mat{x};\mat{W}, \mat{b})$ known as the \textit{network}. The network is parameterized by an unknown weight vector $\mat{W}$ and bias $\mat{b}$. We assume that we can query the network's nonlinear input-output mapping:
\begin{equation}\label{eq:input-output}
    \mat{x} \mapsto f(\mat{x};\mat{W},\mat{b}) \equiv \mat{y}.
\end{equation}
That is, we may apply input stimuli $\mat{x}$ and record the output $\mat{y}$. 
Section \ref{sec:mouse} describes how we collect such data from the mouse retina.

Our goal is to recover the true $\mat{W}$ and $\mat{b}$ given such $(\mat{x},\mat{y})$ queries. 
A first question that arises is whether recovering the true $\mat{W}$ and $\mat{b}$ is possible, even with infinite data---i.e., whether $\mat{W}$ and $\mat{b}$ are uniquely identifiable.  Other questions include how to accurately recover $\mat{W}$ and $\mat{b}$ (or at least their non-zero support) given finite and noisy data, which data points to query, and what types of domain knowledge can aid in this process.  

In practice, estimating the true $\mat{W}$ and $\mat{b}$ given training data is tackled as a regression problem.  We are particularly interested in the case where $\mat{W}$ is sparse. The conventional way to encourage sparsity is to use $\ell_1$-regularization \cite{tibshirani1996regression}, which we will also employ.  We discuss in Section \ref{sec:alg} practical considerations through extensive evaluation of simulated biological circuits.


\textbf{Summary of Domain Knowledge.} 
Our goal is to not only establish conditions where system identification of nonlinear feedforward networks is possible, but also that those conditions be practically relevant.  Guided by neuroscience domain knowledge, we study circuits with the following properties:
\begin{enumerate}[label=(\roman*),itemsep=0.2pt,topsep=0pt]    
    \item The nonlinearity of individual neurons is well understood, and can be well modeled using ReLUs or leaky ReLUs \cite{baccus2002}.
    \item The dominant computation of the neural circuit is feedforward, which is true for circuits found in the retina \cite{wassle1991, gollisch2010}.
    \item Many potential weights can be preemptively set to zero, due to the implausibility of neurons in various spatial configurations being connected to one another \cite{dunn2014convergence, siegert2009addressbook}.
    \item All weights can be sign constrained (either non-negative or non-positive), due to knowing the inhibitory or excitatory behavior of each neuron \cite{wassle1991, gollisch2010, raviola1982}.
\end{enumerate}
Items (i) \& (ii) above imply knowing the functional form (i.e., multi-layer perceptron with known number of layers, maximal number of neurons in each layer, and form of the nonlinearity). Item (iii) implies that one can reduce the number of free parameters in the model, thus easing the burden of learning (although still requiring high-dimensional sparse estimation). Item (iv) is perhaps the most interesting property, as it effectively constrains the weights to be within a single known orthant.  We show in Section \ref{sec:theory-result} that the sign constrained condition is an important sufficient condition for proving identifiability, and we show empirically in Section \ref{sec:alg} that $\ell_1$-regularized regression can succeed in system identification on sign constrained networks and fails on unconstrained networks.

\section{Theoretical Analysis}\label{sec:theory-result}

Motivated by the neuroscience domain knowledge discussed in Section \ref{sec:problem}, we now establish factors
that govern fine-grained identifiability of neural networks---a topic that has been studied since the 1990s \citep{sussman,albertini,fefferman}. To develop the core ideas, it will help to consider the following neural network:
\begin{equation}\label{eq:2-layer}
    f(\mat{x}) := \mat{W_2}\max(\mat{0}, \mat{W_1}\mat{x} + \mat{b_1}) +
  \mat{W_3x} + \mat{b_2},
\end{equation}
for weights $\mat{W_1}\in\R^{n_1\times n_0}$, $\mat{W_2}\in\R^{n_2\times n_1}$, 
$\mat{W_3}\in\R^{n_2\times n_0}$ and biases $\mat{b_1}\in\R^{n_1}$, 
$\mat{b_2}\in\R^{n_2}$.

\begin{figure}
\centering
\includegraphics[width=0.40\textwidth]{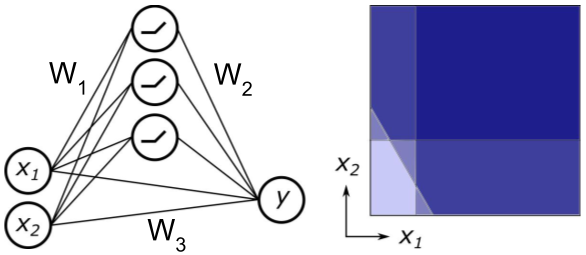}
\vspace{-0.1in}
\caption{Left: A one hidden layer
    ReLU network with skip connection (Equation \ref{eq:2-layer}). Right: The nonlinearity transitions occur on hyperplanes in input space.}
\label{fig:relu_network}
\end{figure}

There are three reasons why this is a sensible network to consider. First, the network's relatively simple structure facilitates theoretical insight. Second, this network bears a resemblance to neural (sub-)circuits found in the retina (see Fig. \ref{fig:bigsim_circuit}). And third, if one ignores the skip connections, biases, and nonlinearity, then the identification problem reduces to factorization of $\mat{W_2W_1}$, allowing us to compare to results on matrix factorization  \cite{donoho_when_2004,laurberg_theorems_2008,nnmf_revisited}. 

We shall now provide a result on identifiability for this class of nonlinear networks. The proof of Theorem \ref{thm:identifiability} is given in Appendix~\ref{sec:proof}.

\begin{restatable}{theorem}{identifiability}
\label{thm:identifiability} 
  Suppose that network (\ref{eq:2-layer}) satisfies the following three conditions:
  \begin{enumerate}[label=(\roman*),itemsep=0.2pt,topsep=0pt]
  \item no column of $\mat{W_2}$ or row of $\mat{W_1}$ is entirely zero;
  \item no two rows of
  $\mat{W_1}$ are collinear;
  \item $\mat{W_1}$
  is nonnegative.
  \end{enumerate}
  Let $\mat{P}$ denote an unknown permutation matrix and $\mat{D}$ an unknown positive diagonal matrix. Then, by input-output queries of the form (\ref{eq:input-output}), we may recover $\mat{DPW_1}$, $\mat{DPb_1}$, $\mat{W_2P^{-1}D^{-1}}$, $\mat{W_3}$ and $\mat{b_2}$.
\end{restatable}
The appearance of permutation matrix $\mat{P}$ reflects the invariance of a two-layer network to permutations of its hidden units. Also, by positive homogeneity of the $\max$ function, the output synapses of a hidden unit (columns of $\mat{W_2}$) may be scaled up by some $\alpha>0$ provided that the input synapses (rows of $\mat{W_1}$) are scaled down by $1/\alpha$. This gives rise to the diagonal matrix $\mat{D}$. These symmetries are innate to two-layer systems---the same issue is present in matrix factorization \citep[Definition~4]{nnmf_revisited}.

This theorem is somewhat typical of the results in the identifiability literature. For example, in 1992 \citet{sussman} proved a similar result for a network with one hidden layer, $\tanh$ nonlinearity, and no skip connections. In contrast, our result applies to networks with skip connections and relu nonlinearity. 

\textbf{Discussion of preconditions and assumptions.}
Condition (i) imposes that every hidden unit must be connected to both the input and output of the circuit, and condition (ii) imposes that no hidden unit is computationally redundant with another. These conditions are intuitively important for identifiability. The most substantive preconditon is condition (iii), which imposes a sign constraint on the synapses at the first layer. In the neuroscience context, this would correspond to knowing that all of the synapses in the first layer are excitatory as is the case for the outgoing synapses of bipolar cells in the retinal circuit (see Fig. \ref{fig:bigsim_circuit}). This demonstrates that domain knowledge of neuronal cell types can simplify both the theoretical analysis as well as the recovery of the neural network's weights. 

Our theoretical analysis has some notable limitations.  For instance, our result requires querying the network on inputs $\mat{x}$ lying in the negative orthant, which is not biologically plausible in cases where $\mat{x}$ is modeling the output of a layer of rectified nonlinear neurons. Additional constraints on the bias $\mat{b_1}$ would be required to guarantee identifiability given only non-negative input queries.  Furthermore, our result makes no comment on sample complexity or dealing with non-convexity of the underlying optimization problem, which are all interesting directions for future work.  We do, however, provide an extensive empirical study of such questions in Section \ref{sec:alg}.

\textbf{Connection to identifiability in other systems.}
The preconditions of Theorem~\ref{thm:identifiability} are mild in comparison to results in the nonnegative matrix factorisation literature, which require strong sparsity conditions on $\mat{W_1}$ to guarantee identifiability of $\mat{W_2W_1}$ \citep{donoho_when_2004,laurberg_theorems_2008,nnmf_revisited}. Theorem \ref{thm:identifiability} suggests that---surprisingly---the presence of nonlinearities can make system identification \textit{easier}. This is because the location of the nonlinear thresholds in input space is what reveals the entries of $\mat{b_1}$ and $\mat{W_1}$ up to a scaled permutation. We illustrate this in Figure \ref{fig:relu_network}.

\section{Empirical Analysis \& Practical Considerations}
\label{sec:alg}

We now empirically analyze conditions under which $\ell_1$-regularized regression can lead to successful fine-grained identification of neural circuits, and cross-reference with the theoretical results from Section \ref{sec:theory-result} where appropriate.
Practical system identification of neural networks can be challenging even with significant constraints imposed (e.g., many weights preemptively constrained to zero). Most applications of artificial neural networks (ANNs) are not concerned with identification of the ``true'' network, but rather focus on predictive performance (which can be achieved by many parameterizations). Therefore, it is important to understand, operationally, when system identification can be reliably performed. 

\begin{figure*}  
 \includegraphics[width=0.7\textwidth]{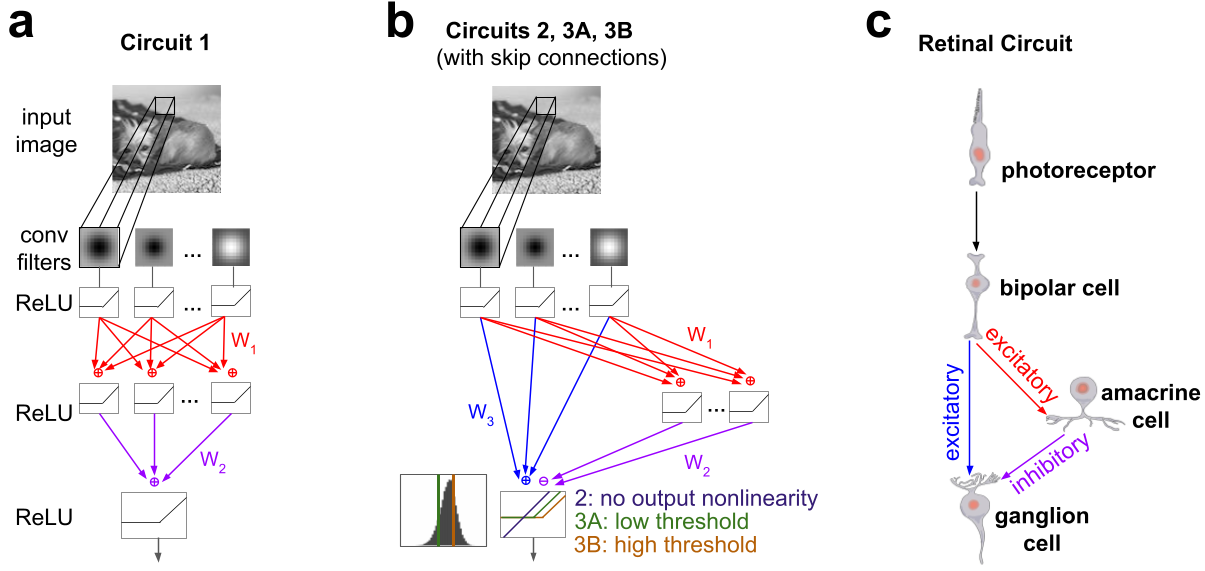}
 \vspace{-0.1in}
   \caption{Networks studied in simulation. \textbf{(a)} Simulated circuit. \textbf{(b)} Three simulated circuits with skip connections, varying in the use of output nonlinearity. \textbf{(c)} The retinal circuit that inspired the models in (a) and (b). The photoreceptor layer is treated as linear, and the convolutional layers in (b) and (c) represent the combined computation of photoreceptors and bipolar cells.}
  \label{fig:bigsim_circuit}
\end{figure*}

Since we assume that the true neural circuit is sparse, we employ $\ell_1$-regularized regression.
Given a training set $S = \{(x,y)\}$, tuning parameter $\lambda$, neural circuit $f(\cdot;\mat{W})$, and set of allowable weights $\W$, the optimization problem can be written as:
\begin{equation}
\mat{W}_k \leftarrow \argmin_{\mat{W}\in \W} \sum_{(\mat{x},y)\in S} \|f(\mat{x};\mat{W}) - y\|^2 + \lambda |\mat{W}|.\label{eq:optimizationproblem1}
\end{equation}
Due to the nonconvex nature of this optimization, it is preferable to run it with multiple ($K$) random initializations, thereby producing a set of solutions $\mat{W}_1,\ldots,\mat{W}_K$.
 
Using controlled simulated settings, we performed a systematic study to understand the dependence of successful system identification on the following factors:
\begin{enumerate}[label=(\roman*),itemsep=0.2pt,topsep=0pt]    
    \item The design of the function class $f(\cdot; \mat{W}, \mat{b})$, and in particular the presence of skip connections in the architecture.
    \item The number of non-zero weights $\mat{W}\in\W$, via preemptively setting many connections to zero.
    \item The use of sign constraints on the weights $\mat{W}\in\W$.
    \item The use of $\ell_1$ regularization to encourage weight sparsity.
    \item The design of the training dataset, i.e., which data points to query.
\end{enumerate}

In each simulated experiment, we created a sparse oracle network, and then trained a second network to recover that structure via the optimization problem in Equation \ref{eq:optimizationproblem1}, while varying the constraints, regularizers and datasets.

We define two measures of success:
\begin{definition}[System Recovery Score]\label{def:sysrs}
Let the {\it system recovery score} be defined on the vectorized weight matrices as:
\begin{equation}\label{eq:system_recovery}
R_{\text{sys}}=2\times\frac{\mat{W_\text{learned}}\cdot \mat{W}_\text{oracle}}{|\mat{W}_\text{learned}|^2+|\mat{W}_\text{oracle}|^2}.
\end{equation}
\end{definition}

\begin{definition}[Support Recovery Score]\label{def:suprs} We will denote by $\overline{\mat{W}}$ the vector of the same size as $\mat{W}$ where $\overline{\mat{W}}_i=0$ if $\mat{W_i}=0$ and 1 otherwise. Let the {\it support recovery score} be defined on the vectorized and binarized weight matrices as:
\begin{equation}\label{eq:support_recovery}
R_{\text{supp}}=2 \times \frac{\overline{\mat{W}}_\text{learned}\cdot \overline{\mat{W}}_\text{oracle}}{|\overline{\mat{W}}_\text{learned}|^2+|\overline{\mat{W}}_\text{oracle}|^2}.
\end{equation}
\end{definition}

Both measures are related to cosine similarity and the Jaccard coefficient.  The System Recovery Score achieves perfect performance when the exact system is recovered with the exact weights.  The Support Recover Score achieves perfect performance whenever the non-zero support (i.e., the structure) is recovered exactly. 

\subsection{Summary of findings} 
Fig. \ref{fig:bigsim} shows our main findings on the practicality of the approach.  For all simulated circuits, we are able to reliably recover most of the structure and parameters using $\ell_1$-regularized regression, when there are around 500 non-zero, sign-constrained free parameters in the network.
 As we will see in Section \ref{sec:mouse}, the biological problems we are interested in have a few hundred free parameters, and so are closely modeled by the darkest curves in Figures \ref{fig:bigsim} and \ref{fig:failure_modes}. 
 
 Fig. \ref{fig:failure_modes} shows that successful recovery can be contingent on the use of $\ell_1$-regularization as well as sign constraints.  In particular, we see that $\ell_1$-regularization is especially crucial when the circuit has skip connections (as is common in the retina), and that the estimation problem fails completely without sign constraints.  The former finding potentially suggests skip connections as an important factor for finite-sample analysis of system identification of nonlinear circuits, and the latter finding supports our theoretical finding that sign constraints are an important condition for guaranteeing identifiability.

\subsection{Network architecture}
We focus on the most generic network architectures that are prevalent in retinal biology---though similar architectures arise in gene regulatory networks \cite{mccarter2014sparse}. We call these networks: 
\begin{itemize}
    \item Circuit 1:  three layers, with the first layer convolutional and fixed, and the output layer nonlinear (Fig \ref{fig:bigsim_circuit}a).
    \item Circuit 2: three layers, with the first layer convolutional and
fixed, the output layer linear, including skip
connections (Fig \ref{fig:bigsim_circuit}b).
    \item Circuit 3A: Circuit 2 with the addition of a low-threshold ReLU output nonlinearity (Fig \ref{fig:bigsim_circuit}b).
    \item Circuit 3B: Circuit 2 with the addition of a high-threshold ReLU output nonlinearity (Fig \ref{fig:bigsim_circuit}b).
\end{itemize}
Note that Circuit 2 is, in fact, the circuit studied theoretically in Section \ref{sec:theory-result}, but with a restricted input domain. In the retina, one can think of the output of the first layer as modeling bipolar cell glutamate release rates, the second layer as modeling amacrine cells, and the third as modeling ganglion cells, though this ``feedforward loop'' circuit motif also appears in other domains, most prominently in transcriptional regulatory networks \cite{alon2020sysbio, gui2016}.

For each of the four architectures, we constructed a specific oracle circuit by selecting a fixed set of weights. Circuit 1 had 36 nonzero synapses. Circuits 2, 3A and 3B had 96 nonzero synapses. For each of these four circuits, we generated a training set by presenting images to the circuit and recording the output responses. 
The rest of Section \ref{sec:alg} will analyze recovery performance of these architectures while varying the other factors of interest.

\subsection{Non-zero connection constraints are helpful for identification of some architectures}
\label{sec:num_params}

We first considered varying the number of free parameters in the ANN.  Intuitively, with more free parameters, it will be harder to recover the true network.  However, constraining to fewer parameters would require increasing amounts of domain knowledge.
%
%
%

Existing neuroscience domain knowledge constrains many connections to zero weight. In the retina, such connectivity constraints are implemented by the precisely organized anatomy of neurons. The so-called inner plexiform layer (IPL) is a meshwork of synapses between different unit types (bipolar, amacrine, and ganglion cells). Each type sends its axons and dendrites into a distinct lamina of the IPL, and neurons may connect only if they co-stratify in at least one lamina. This implements an  ``address book'' of allowable connectivity between cell types in the network \cite{siegert2009addressbook}. Using anatomical studies in the literature \cite{zhang2012oms, ghosh2003, bae2017eyewire, tsukamoto2017, franke2017, greene2016, kim2014ds, munch2009}, we have compiled such an address book for 36 retinal cell types (see appendix). A strong address book constraint translates to fewer free parameters.

For convenience, the first layer of weights ($W_1$) were also constrained so that each amacrine cell analogue in layer 2 only received input from a restricted spatial region of the output of layer 1. This is to remove permutation symmetry of $W_1$ and $W_2$, and it is also realistic, since it is known that amacrine cells in the retina typically pool inputs over a small region in space, not the entire visual field. Scaling symmetry of the weights was removed by fixing biases in place during training (see A.1.4 for details.)

Fig. \ref{fig:bigsim} shows our main results, where we also constrain the sign of each weight based on whether the corresponding biological synapse should be excitatory or inhibitory.  We see that 
all circuits can be approximately identified ($R_\text{sys} \geq 0.9$) given enough training data and if sufficiently constrained in the number of free parameters (around 500).  
Circuit 1, which is the only circuit without skip connections, always yields high system and support recovery (given enough training data) irrespective of the  number of free parameters.  However, the successful recovery of Circuits 2, 3A, and 3B exhibits a strong dependence on the number of free parameters. Our theory directly implies that Circuit 2 should be identifiable, which suggests that the non-convex optimization landscape of ANN training and the restricted input domain may be a significant factor here.

\begin{figure}  
\centering
 \includegraphics[width=0.49\textwidth]{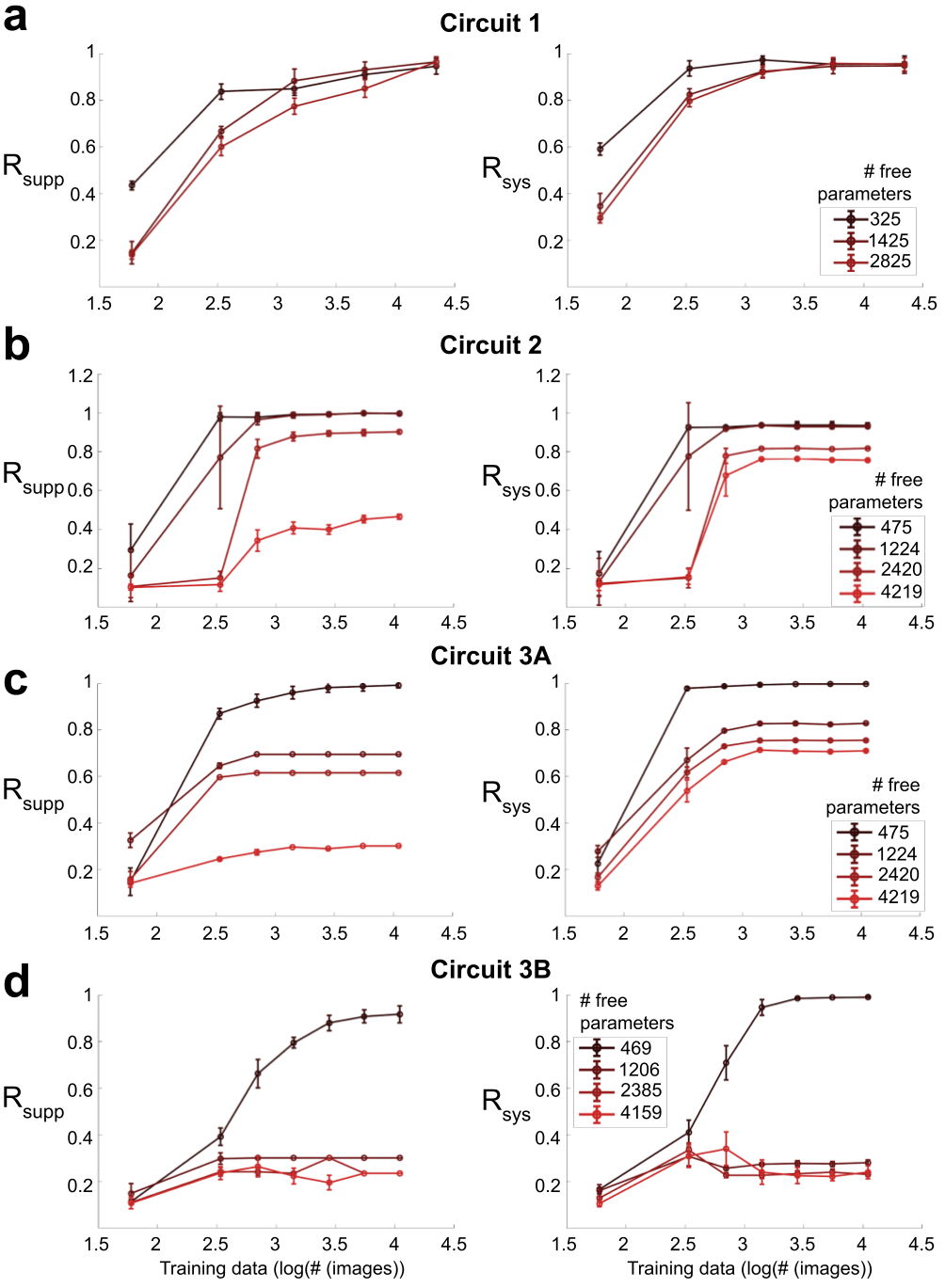}
 \vspace{-0.21in}
   \caption{System and support recovery scores (Definitions \ref{def:sysrs} and \ref{def:suprs}) for Circuits 1, 2, 3A and 3B respectively, using $\ell_1$-regularization and sign constraints. Each curve represents an ANN with 325--4219 free parameters (see legend). 
   The ANN was given varying quantities of free parameters (weights associated with synapses) and training data. For each circuit structure, a single training dataset was generated, then networks with varying numbers of weights at initialization were given access to some fraction of the training set and trained. System and support recovery scores were computed in each case. Mean and standard deviation of these scores are plotted for 10 training runs, each with a different random initialization. 
  One can reliably recover most of the system structure and parameters when there are around 500 free parameters (corresponding to putative synapses) in the system.
   }
  \label{fig:bigsim}
 \vspace{-0.05in}
\end{figure}

\subsection{Sign constraints are necessary for system identification}

Another way to constrain the network is to use sign constraints, due to abundant neuroscience domain knowledge about whether certain neurons are either excitatory or inhibitory \cite{wassle1991, gollisch2010, raviola1982}. In fact, for the retinal circuits we study in Section \ref{sec:mouse}, we have sufficient knowledge to sign constrain every weight.
Fig. \ref{fig:failure_modes} c\&d shows the results of an ablation study that removes the sign constraints.  We see a sharp contrast compared to Fig. \ref{fig:bigsim} in that removing the sign constraints causes system identification via \eqref{eq:optimizationproblem1} to fail.  
This finding is consistent with our theoretical exploration and Theorem \ref{thm:identifiability}, where the sign constraint was a crucial element in identifying the half planes generated by the ReLU units.

\begin{figure} 
 \includegraphics[width=0.45\textwidth]{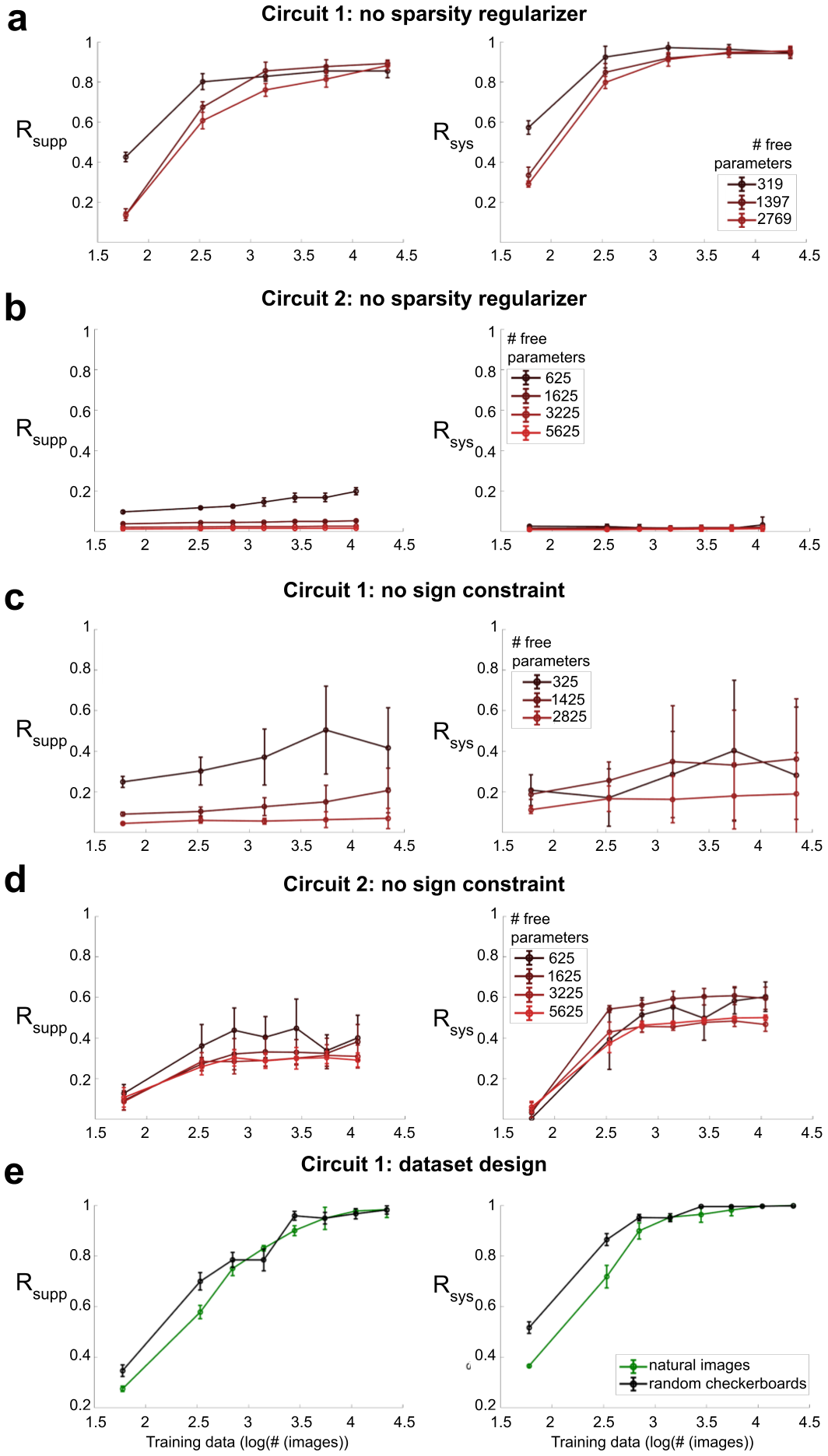}
 \vspace{-0.15in}
  \caption{\textbf{Results of ablation studies.  Successful system identification can be contingent on regularization, sign constraint and dataset type.} System identification for simulated retina-style circuits:
    {(a)} for Circuit 1 proceeds as normal without $\ell_1$-regularization; (b) for Circuit 2 is impaired without $\ell_1$-regularization; (c) for Circuit 1 is impaired without the sign constraint; (d) for Circuit 2 is impaired without the sign constraint.  In (e) we see that random noise images lead to better system identification of Circuit 1 than natural images when data is scarce.
  } 
  \vspace{-0.2 in}
  \label{fig:failure_modes}
\end{figure}

\subsection{Sparsity regularization is helpful for identification of some architectures}
 We next investigate the benefits of using $\ell_1$-regularization in practice.  Fig. \ref{fig:failure_modes} a\&b show the results of an ablation study that omits the use of $\ell_1$-regularization.  Compared to Fig. \ref{fig:bigsim}, we see that regularization never hurts system identification, and sometimes appears to be crucial in practice.  In particular, one can still perform identification without $\ell_1$-regularization for Circuit 1, which has no skip connections, but not for Circuit 2, which has skip connections.  We hypothesize that a theoretical characterization of sample-efficient system identification of neural circuits could depend in a non-trivial way on the presence or absence of skip connections.
 
 
 \subsection{Dataset design matters when data is scarce}
In visual neuroscience experiments one must choose among many possible images or movies with which to query the neural circuit to collect training data. Practically, biological experiments are time-limited, and therefore, the training dataset size is limited. Unsurprisingly, system identification of all four circuits showed a strong dependence on the size of the training dataset (Fig \ref{fig:bigsim}a-d).

Our final analysis in this section is to compare the efficacy of two very different data collection approaches: white-noise images vs natural photographs. White noise has a long history for system identification in engineering, whereas natural images better reflect the signal domain that the retina evolved to handle. White-noise stimuli were generated as random checkerboards. Each image was $100 \times 100$ pixels, and consisted of $5 \times 5$ grayscale checkers of random intensity. Natural grayscale images were taken from the COCO (Common Objects in Context) dataset \cite{lin2014coco}. An overparameterized artificial neural network with architecture matching Circuit 1 (Fig. \ref{fig:bigsim_circuit}a) was trained with varying quantities of data from both sets. When training data were plentiful, performance on the two datasets was comparable. However, when data were more scarce, random checkerboard stimuli led to better system identification (Fig. \ref{fig:failure_modes}e). Interpolating, we find that a system recovery score of at least $R_{\text{sys}}= 0.9$ can be achieved with about 400 checkerboard images, but requires 630 natural images.


\section{Case study on the mouse retina}
\label{sec:mouse}


\begin{figure*}
 \includegraphics[width=\textwidth]{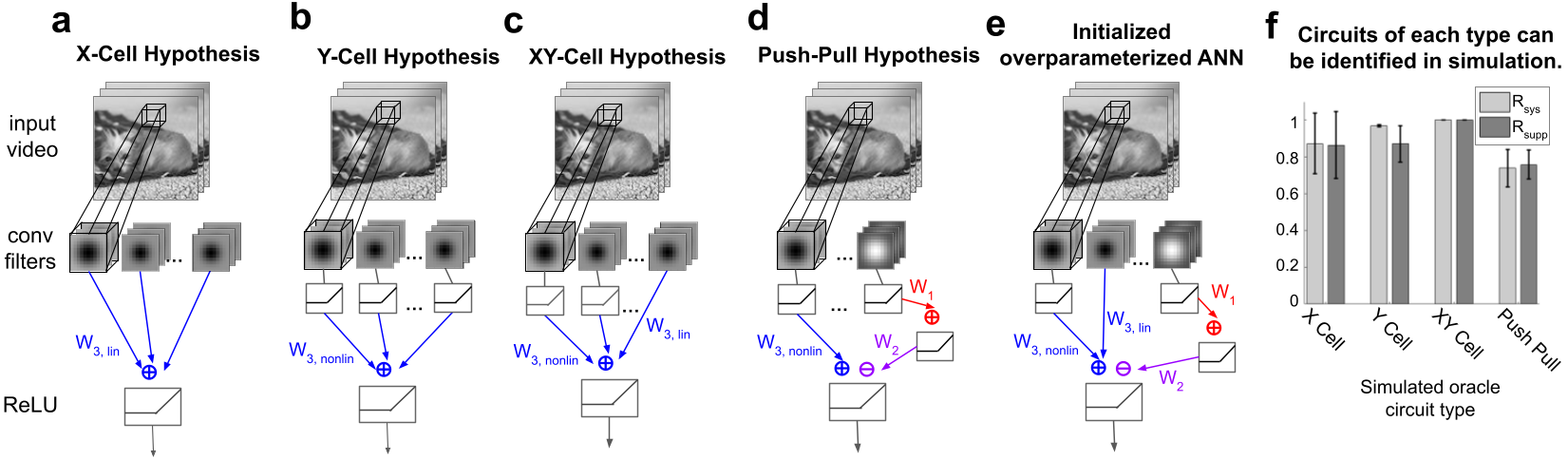}
 \vspace{-0.7cm}
  \caption{Hypothesized circuit architectures for the alpha ganglion cell. Each structure convolves the input video with a bank of spatiotemporal filters in the first layer. {\textbf (a)} X-Cell circuit: utilizes only linear units in the first layer. {\textbf (b)} Y-Cell circuit: utilizes only nonlinear units in the first layer. {\textbf (c)} XY-Cell circuit: utilizes both linear and nonlinear units in the first layer. {\textbf (d)} Push-pull circuit: utilizes nonlinear units in the first layer. Includes an additional inhibitory pathway. Structurally analogous to circuits 3A and 3B in Fig. \ref{fig:bigsim_circuit}b. (e) The ANN is initialized to include all components and connections used in all four hypotheses. (f) Structure and support recovery scores when the ANN in (e) is trained on data generated by simulated circuits with architecture matching (a)-(d).}
  \label{fig:realdata_circuit}
\end{figure*}

We conclude this paper with a case study on  the mouse retina.  Our goal here is to demonstrate the practicality of fine-grained neural system identification in real biology settings.  In many cases, one does not have the actual ground truth neural circuit, and so we evaluate by confirming that we can recover models that are consistent with known neuroscience domain knowledge.

We specifically studied four types of mouse retinal ganglion circuits:  the sustained OFF, sustained ON, transient OFF, and transient ON subtypes of the alpha retinal ganglion cell (sOFF$\alpha$, sON$\alpha$, tOFF$\alpha$, and tON$\alpha$ RGCs). The circuitry of interneurons that provide input to these cells is not yet known in retinal biology. However, there is strong evidence that the tOFF$\alpha$ cell pools excitatory input from OFF bipolar cells and inhibitory input from ON bipolar cells via amacrine cells \cite{munch2009}. These cell types present an interesting challenge for system identification. Namely: can we use this technique to (1) confirm what is known about the tOFF$\alpha$ cell circuit (2) provide new specific circuit hypotheses for further study? 


We studied four hypothetical models of the $\alpha$ cell based on current domain understanding: the X-Cell Model, the Y-Cell Model, the XY-Cell Model, and the Push-Pull Model model (Fig. \ref{fig:realdata_circuit}a-d). These hypotheses are specific to alpha cell physiology, and therefore include some variations on the more generic architectures studied in Section \ref{sec:alg}. The X- and Y-Cell models were inspired by retinal circuits commonly described in the literature \cite{enrothcugell1966}. They are characterized by their use of linear or nonlinear units in layer 1 respectively. These units correspond to bipolar cells, and it is known that some bipolar cells report light intensity linearly, while others do so nonlinearly \cite{baccus2002}. The XY-Cell model is a hybrid of these two, and the first layer of this circuit includes both linear and nonlinear units. The Push-Pull model uses both excitation from nonlinear units as well as inhibition from other nonlinear units through an intermediate nonlinear unit (structurally similar to Circuits 3A and 3B in Fig. \ref{fig:bigsim_circuit}b).

\subsection{Additional practical considerations}
Leveraging insights from Section \ref{sec:alg}, we designed the learning problem to maximize the probability of success. The architecture of the artificial neural network (ANN) was designed to include the necessary synapse types for each of the four hypothesized circuits. The ANN was initialized as a superposition of all components and connections included in the X-Cell, Y-Cell, XY-Cell, and Push-Pull models. It therefore had the opportunity to converge to one of the four hypotheses, any hybrid thereof, or a completely different circuit. We initialized the ANN with enough artificial ``bipolar'' and ``amacrine'' cell units to tile the entire stimulus (though we know that the true circuit will only tile a small portion of it) while adhering to known properties of retinal cells including receptive field size and approximate synaptic convergence ratios \cite{dunn2014convergence}. The reason for this design choice is to ensure that the model has the capacity to contain the true circuit. This network had 417 free parameters (comparable to the darkest curves in Figures \ref{fig:bigsim} and \ref{fig:failure_modes}). 

The ANN also included a temporal processing component to account for the temporal correlations inherent to ganglion cell responses. The stimuli used were temporal sequences of white noise images and the response collected from the ganglion cells was in the form of firing rate over time (Fig. \ref{fig:realdata_results}b). The training sets were comprised of 4500-7000 2-second videos and corresponding ganglion cell responses. Our results from Section \ref{sec:alg} suggest that, under these conditions, most of the circuit should be recovered.

\subsection{Sanity check simulation: each circuit structure can be correctly identified}
 As a preliminary test, each of the four structural models (Y-cell, X-cell, XY-cell and Push-pull) was simulated. A training set of ganglion cell firing rates in response to visual stimuli was generated. An ANN initialized to include all connections necessary to replicate any of the four models, as well as various other ``in-between'' structures (Fig. \ref{fig:realdata_circuit}e), was trained in turn on each training set. It correctly selected the appropriate circuit hypothesis each time it was trained, regardless of random initialization, and achieved high values of $R_{\text sys}$ and $R_{\text supp}$ in every case (Fig. \ref{fig:realdata_circuit}f). By comparison, $R_{\text sys}$ scores for hypotheses different from the model that generated the data were considerably lower (Fig. \ref{fig:intramodel}). We therefore  conclude that one can reliably identify which structural model gave rise to a dataset. 

\subsection{System identification with real data confirms the known circuitry of the \texorpdfstring{tOFF$\alpha$} cell and suggests a new pathway}


Prior research in neuroscience suggests that the tOFF$\alpha$ cell uses a Push-Pull motif in its circuitry \cite{munch2009}. Training the ANN described above (Fig. \ref{fig:realdata_circuit}e) on data from mouse tOFF$\alpha$ cells corroborates this. The network was trained on data from 14 tOFF$\alpha$ cells with 5-10 random initializations for each cell. In every case, the resulting structure was not of any of the four types hypothesized, but instead a hybrid we call the XY-Push-Pull type. Thus the neural network replicated the results of the biological experiments performed in \cite{munch2009}, but also suggested an additional pathway, namely, the inclusion of linear OFF bipolar cells which excite the tOFF$\alpha$ ganglion cell. Training the same network on data from other alpha cell types suggested the same circuit motif, but with varying relative weights. Further biological experiments could be conducted to confirm this finding.  Thus, fine-grained identification of neural circuits can serve as a guide for these experimental efforts, allowing a biologist to prioritize hypotheses and direct future investments.

\begin{figure}
\centering
 \includegraphics[width=0.49\textwidth]{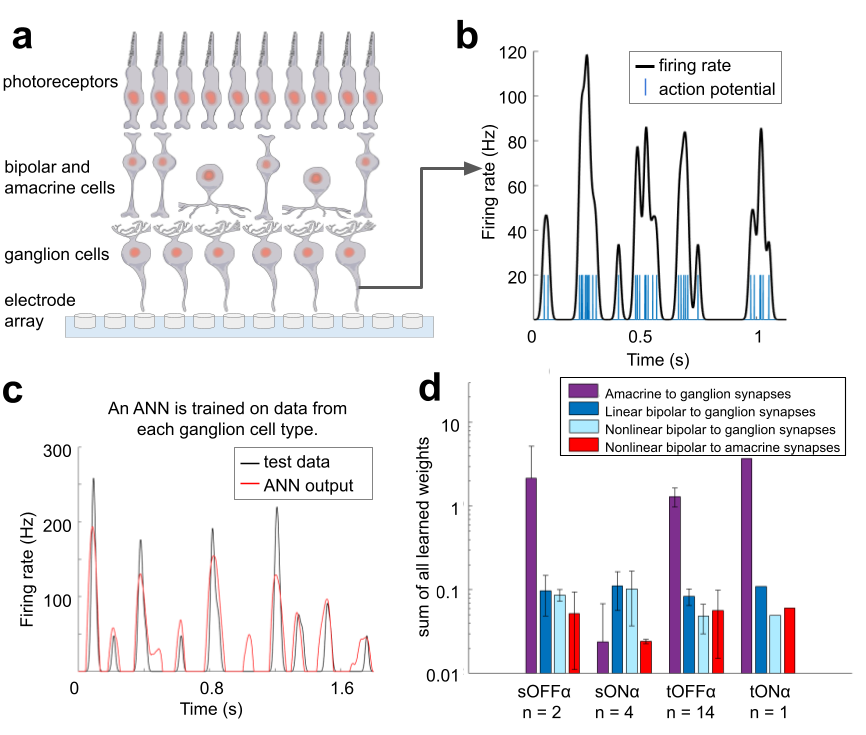}
  \vspace{-0.27in}
   \caption{\textbf (a) Schematic of mouse retina on microelectrode array. The retina is positioned so that the ganglion cell layer is flush with the array of electrodes. Hundreds of ganglion cells can be simultaneously recorded. (b) Example of signal from a single mouse retinal ganglion cell. Blue ticks are single action potentials. Black curve shows the smoothed firing rate of the neuron over time. (c) Example of quality of fit to ganglion cell firing rate. (d) Summed weights for artificial neural networks trained on datasets recorded from each of the four ganglion cell types in mouse retina.}
  \label{fig:realdata_results}
\end{figure}


\section{Conclusion}
We have demonstrated the potential to perform accurate fine-grained system identification of nonlinear neural circuits. 
It can both validate existing understanding of a biological system, as well as put forth hypotheses for biological circuitry. Future theoretical directions include studying finite-sample guarantees and understanding the dependence of such guarantees on skip connections.  Future modeling directions include combining with biophysical models \cite{schroderneurips} into a unified framework.  %
More generally, this method of system identification may serve as a useful tool for scientists aiming to gain microcircuit-level understanding of systems about which only coarse-grained information is known.

Biology has seen an explosion of technological development in the past decade. Alongside this advancement in electrophysiological technologies and anatomical tools, we will need to put forth new, more powerful computational techniques that can convert large datasets into biological insights. While the retina is still among the best-understood regions of the brain, we anticipate that this level of understanding may soon be gained about other biological systems, as well as in other scientific domains. As this begins to unfold, system identification using neural networks will serve as a powerful tool for hypothesis generation and elimination. 
Rather than develop highly specialized algorithms for each specific application, perhaps accessible deep learning tools can be combined with regularizing domain knowledge, and used more commonly by scientists to uncover the secrets of nonlinear circuits.









\begin{acks}
This material is based upon work supported by the National Science Foundation Graduate Research Fellowship under Grant No. 1745301. Any opinion, findings, and conclusions or recommendations expressed in this material are those of the author(s) and do not necessarily reflect the views of the National Science Foundation. This work was also funded by a cloud computing grant from Amazon Web Services in collaboration with the Information Science and Technology initiative at Caltech. This work was also supported by the Simons Collaboration on the Global Brain (grant 543015 to Markus Meister).  Jeremy Bernstein was supported in part by an NVIDIA Fellowship and by NASA TRISH-RFA-BRASH 1901. Yu-Li Ni was supported by Taipei Veterans General Hospital-National Yang-Ming University Excellent Physician Scientists Cultivation Program, No. 103-Y-A-003. The authors thank James Parkin for providing original illustrations of retinal neurons.
\end{acks}

\bibliographystyle{ACM-Reference-Format}
\bibliography{refs}

\pagebreak

\appendix
\section{Appendix}

\subsection{Reproducibility: Empirical study}

\subsubsection{Computing methods and data repository}
All code for this project was run on Amazon Web Services p2.xlarge instances, equipped with the Deep Learning AMI. All trainings and simulations were run using Tensorflow v1.15.0 \cite{tensorflow2015-whitepaper} and Python 2.7. All trainings were done using the Adam optimizer \cite{kingma2014adam}. Precise hyperparameters used to generate the figures in this paper can be found in the KDD\_readme file. Training code and evaluation code can be found at \url{https://github.com/dbagherian/neural-sys-id}. Data files and examples of trained models can be found at \url{https://doi.org/10.22002/D1.1989}. 

All hyperparameters were selected via cross validation. During cross-validation, the learning rate was varied between $10^{-4}$ and $10^{-2}$, while the $\ell_1$-regularizer strength parameter, $\lambda$, was varied between $10^{-7}$ and $10^{2}$. The README file in the link above lists all hyperparameter values needed to recreate every figure in this paper.

\subsubsection{Complexity and runtimes}
The complexity of system identification can be understood in units of ``network training runs.'' The number of free parameters in each network varied between 300 and 6000. Each training in figures \ref{fig:bigsim}, \ref{fig:failure_modes}, \ref{fig:realdata_circuit} and \ref{fig:realdata_results} was run 5-10 times with different random initializations. Average runtime for these simulated trainings varied between 5 minutes and 4 hours, depending on the quantity of training data used.


\subsubsection{Datasets}
In the simulations in section \ref{sec:alg}, between 60 and 22000 training examples were used, and between 2000 and 2080 examples were used as a test set. Stimuli used for training were of two forms: random checkerboard and natural images. Random checkerboard stimuli were 100$\times$100 pixel images with 5$\times$5 pixel grayscale checkers of continuously varying intensity in the empirical simulations. Natural images were derived from the COCO (Common Objects in Context) dataset \cite{lin2014coco}. Images from this set were made grayscale and randomly cropped to 100$\times$100 pixels.

\subsubsection{Additional details of empirical simulations}
The simulated biological circuits studied in section \ref{sec:alg} were inspired by the retina in multiple ways. In the first layer of these circuits, spatial convolutional filters are applied to input images. These filters were inspired by retinal spatial receptive fields with antagonistic surrounds \cite{meister1999}, and composed of two opposing 2D Gaussian functions. Three different convolutional filters were used in the oracle networks with Gaussians of varying diameters. 28 different convolutional filters were included in the ANN. 14 of these were OFF filters, activated by dark spots in the image, and 14 were ON filters, activated by light spots in the image. 

The remaining two layers took linear combinations of the units in the prior layer. All layers included ReLU nonlinearities. The output of the network was a continuously-varying ``firing rate'' value for the single output neuron, modeling the rate at which the mouse retinal ganglion cell fires action potentials.

 In order to remove some degeneracy from the space of circuit models being searched, it was necessary to fix the biases in place during training. The toy example in Fig. \ref{fig:bias_fix} illustrates the need for this. In order to efficiently compute the structure recovery score, $R$, we needed to remove this type of degeneracy from the space of circuit structures being searched. When $b$ is fixed in place during training, this degeneracy is removed.
 
 \begin{figure}
 \centering
 \includegraphics[width=0.3\textwidth]{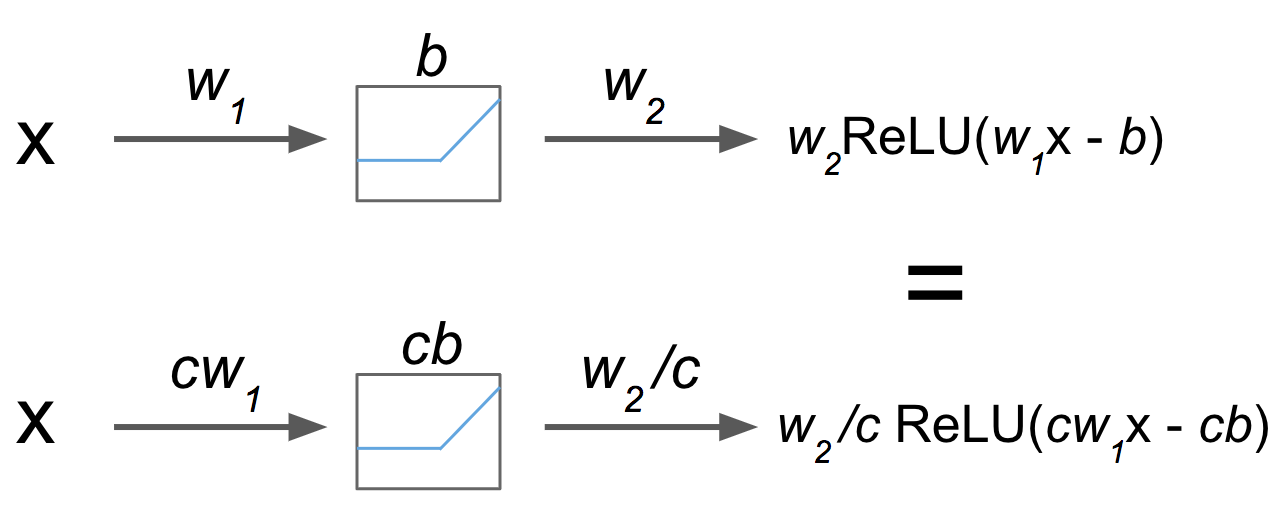}
 \vspace{-0.2 in}
  \caption{{Two toy networks with differing parameter sets which produce the same output for an input, $x$.} }
  \vspace{-0.2 in}
  \label{fig:bias_fix}
\end{figure}

\begin{figure}

\begin{center}
\begin{tabular}{| c | c | c | c | c|} 
\hline
 & X-Cell & Y-Cell & XY-Cell & Push-Pull \\ 
 \hline
 X-Cell & 1.00 & 0.00 & 0.80 & 0.00 \\ 
 \hline
 Y-Cell & 0.00 & 1.00 & 0.40 & 0.07 \\
 \hline
 XY-Cell & 0.80 & 0.40 & 1.00 & 0.04 \\
 \hline
 Push-Pull & 0.00 & 0.07 & 0.04 & 1.00 \\
 \hline
\end{tabular}
\end{center}
\caption{For each simulated circuit used in the simulations in figure \ref{fig:realdata_circuit}, the  $R_\text{sys}$ was computed with every other circuit hypothesis. This serves as a ``floor'' to better contextualize the $R_\text{sys}$ of the trained ANN's in each case.}
\vspace{-0.5cm}
\label{fig:intramodel}
\end{figure}

\subsection{Reproducibility: Mouse retina case study}

\subsubsection{Datasets}
In section \ref{sec:mouse}.2, training data consisted of simulated responses to 60 Hz time-varying randomly flickering bar stimuli with varying degrees of spatial correlation. 85 minutes of simulated data were used to train, and 13 minutes to compute out-of-sample loss. (Bipolar cells in the retina integrate stimulus information in a time window of about 0.2 seconds \cite{baccus2002}.) In section \ref{sec:mouse}.3, about 3 hours of electrophysiological data were used in each training set. 4.5 min of data were used to compute out-of-sample loss.  To collect these data, wild type C57/Bl6 mice were euthanized with cervical dislocation. The retina was dissected out from the eye in an oxygenated nutrient bath, then mounted on a 256-channel electrode array. Visual stimuli were projected onto the retina, while action potentials were recorded extracellularly from the ganglion cell layer \cite{meister1994mea}. We were able to keep the retina alive for 2-5 hours and to record training data in that time. Spike sorting was performed with Kilosort \cite{pachitariu2016kilosort}. $\alpha$ cells were identified post hoc by their responses to a battery of visual stimuli. Training data consisted of responses to time-varying 1D random noise stimuli with varying degrees of spatial correlation. Spike trains were convolved with a Gaussian filter to produce a continuously varying firing rate (Fig. \ref{fig:realdata_results}b). The trained ANN explained a large fraction of the variance of each dataset. By cell type, these fractions are: sOFF$\alpha$ 0.75 $\pm$ 0.04, sON$\alpha$ 0.82 $\pm$ 0.17, tOFF$\alpha$ 0.89 $\pm$ 0.11, tON$\alpha$ 0.65 $\pm$ 0.00.

 \subsubsection{Parameter reduction in the first convolutional layer}
In order to reduce the number of learned parameters and speed up training on the mouse data, we pulled the convolution out of the first layer of the trained network in the following way.

Since the functional form of bipolar cell spatiotemporal filters is well understood in biology \cite{baccus2002, real2017, franke2017, keat2001}, one can parameterize the temporal convolutional filters in the first layer of the ANN using a basis set of stretched sinusoid functions of the form:
\begin{equation}
T_j = \left\{
\begin{array}{ll}
      \sin\Big(\pi j\Big(2 \frac{t}{\tau}-\Big(\frac{t}{\tau}\Big)^2\Big)\Big), &  \text{ for } 0\leq t \leq \tau \\
      0 & \text{ otherwise } \\
\end{array} 
\right. 
\end{equation}

for $j = 1, 2, ..., 16$ \cite{keat2001}.
Let $s_i$ denote the $i$th stimulus video in the training set. The result of the first layer convolution could therefore be rewritten as:

\begin{equation}
    y^1_i = \sum_{j=1}^{16} \beta_j(s_i \ast T_j)
\end{equation}

where $\ast$ denotes the convolution operation. The $\beta_j$ could then be initialized in a biologically realistic regime, derived from bipolar cell electrophysiological recordings \cite{franke2017}, but allowed to vary slightly as the network fit the data. The input to the ANN, therefore, can be thought of as $\{s_i \ast T_j\}_{j=1}^{16}$, rather than $s_i$. This removes the first layer convolution from the backpropagation during training and reduces the number of parameters to learn.

\subsection{Retinal address book}
We constructed a ``retinal address book'' based on an idea first presented in \cite{siegert2009addressbook}. See section \ref{sec:alg} for details.
\balance
Because these anatomical data come from a number of studies with varying forms and thoroughness of anatomical data published, this address book was constructed by painstaking manual inspection. For example, different authors divided the IPL into different numbers of sublaminae, and set their boundaries in different locations. 

We selected a sublamination scheme that was compatible with each of these. In this scheme, the mouse IPL was divided into six sublaminae: outer marginal (0.0-0.28 normalized depth), outer central (0.28-0.47 normalized depth), inner central (0.47-0.65 normalized depth), inner marginal (0.65-1.0 normalized depth) and the ON and OFF ChAT bands (defined by limits of choline acetyltransferase (ChAT) expression). With this sublamination scheme, we were able to neatly define a binary stratification profile for each cell type under study, and to therefore create a binary address book You can view this address book at \url{https://www.dropbox.com/s/tryk5vo856ny42b/IPL\%20Address\%20Book.png?dl=0}


One could imagine, however, a future in which stratification profiles are given in units of normalized IPL depth, with confidence intervals or some other measure of uncertainty, for every retinal cell type. From such a dataset, one could replace the binary entries of this table with a continuously varying measure of co-stratification and uncertainty.
\subsection{Proof of Theorem \ref{thm:identifiability}}
\label{sec:proof}

In this section, we prove the identifiability result of Section  \ref{sec:theory-result}. See \citep{sussman,albertini,fefferman,helmut, rolnick} for related prior work. Recall the definition of network (\ref{eq:2-layer}):
\begin{equation*}
    f(\mat{x}) := \mat{W_2}\max(\mat{0}, \mat{W_1}\mat{x} + \mat{b_1}) +
  \mat{W_3x} + \mat{b_2},
\end{equation*}
for weights $\mat{W_1}\in\R^{n_1\times n_0}$, $\mat{W_2}\in\R^{n_2\times n_1}$, 
$\mat{W_3}\in\R^{n_2\times n_0}$ and biases $\mat{b_1}\in\R^{n_1}$, 
$\mat{b_2}\in\R^{n_2}$. 

\identifiability*
\begin{proof}
    To begin, observe that each component of the vector output $f(\mat{x})$ is a piecewise linear function of $\mat{x}$. We will first show that the transitions of the $\max$ function may be identified with the transitions between linear regions, thereby allowing identification of $\mat{W_1}$ and $\mat{b_1}$.
    
    Since condition (i) excludes the possibility that any row of $\mat{W_1}$ is entirely zero, every row of $\mat{W_1}$ will cause a transition of the max function. For the $j$th row of $\mat{W_1}$, the transition occurs on a hyperplane $H_j$ in the input domain:
    $$H_j := \left\{\mat{x} :  \sum_{k=1}^{n_0} \mat{W_1}^{(jk)}\mat{x}^{(k)}+\mat{b}^{(j)} = 0\right\} \subset \R^{n_0}.$$
    
    The transition of the $\max$ function corresponding to the $j$th row of $\mat{W_1}$ will only affect output component $f^{(i)}$ if $\mat{W_2}^{(ij)}\neq0$. We know by condition (i) that such an index $i$ exists since the $j$th column of $\mat{W_2}$ is not entirely zero. Therefore all rows of $\mat{W_1}$ correspond to a nonlinear transition in at least one output component of $f$. Since the rows of $\mat{W_1}$ are not collinear by condition (ii), we are now sure that the boundaries between linear regions of $f$ correspond to $n_1$ distinct hyperplanes that partition input space. Since the formula of a hyperplane is unique up to scalings, we may recover $\mat{b_1}$ and the rows of $\mat{W_1}$ up to scalings. Since we do not know in which order the hyperplanes should be listed, the recovery is only unique up to scalings \textit{and} permutations.
    
    Since we have identified $\mat{W_1}$ and $\mat{b_1}$ up to scaled permutations of the rows of $\mat{W_1}$, we may now query $f$ in the region of input space where all components of $\max$ return zero. This region surely exists since by conditions (i) and (iii) combined, all rows of $\mat{W_1}$ point into the positive orthant. Therefore, for $\mat{x}$ sufficiently far into the negative orthant, we have that:
    $$f(\mat{x}) = \mat{W_3x} + \mat{b_2}.$$
    The gradient $\nabla_{\mat{x}}f(\mat{x})$ in this region identifies $\mat{W_3}$, at which point $\mat{b}_2$ may be identified via:
    $$\mat{b_2} = f(\mat{x}) - \mat{W_3x}.$$
    All that remains is to identify $\mat{W}_2$. Consider a point $\mat{x^*}$ on the $j$th hyperplane $H_j$ that is far away from the other hyperplanes $\{H_{i}\}_{i \neq j}$. Such a point exists by condition (ii). Let $\mat{x^+}$ and $\mat{x^-}$ denote points in the local neighbourhood of $\mat{x^*}$ but on the positive and negative side of $H_j$, respectively. Observe that:
    $$\frac{\partial f^{(i)}}{\partial \mat{x}^{(k)}}(\mat{x^+})-\frac{\partial f^{(i)}}{\partial \mat{x}^{(k)}}(\mat{x^-})=\mat{W_2}^{(ij)}\mat{W_1}^{(jk)}.$$
    
    Therefore $\mat{W}_2^{(ij)}$ may be identified by measuring the change in gradient of  $f^{(i)}$ across the $j$th hyperplane $H_j$. Of course, since the rows of $\mat{W_1}$ are known only up to a permutation and scale, we will be unsure of the columns of $\mat{W_2}$ up to the same symmetries. 
    
    With $\mat{W_3}$ and $\mat{b_2}$ identified exactly, and $\mat{W_2}$, $\mat{W_1}$ and $\mat{b_1}$ identified up to a permutation and scale, we are done.
\end{proof}

\end{document}